%% file: JAC10.tex
\documentclass{jac}

\include{kkmacros}

\begin{document}

\title{The Block Neighborhood}

\author[lab1]{P. Arrighi}{Pablo Arrighi}
\address[lab1]{Universit\'e de Grenoble, LIG
  \\ 220 rue de la chimie
  \\ 38400 Saint-Martin-d'H\`eres, France
}  
\email{pablo.arrighi@imag.fr}
\urladdr{http://membres-lig.imag.fr/arrighi/}

\author[lab2]{V. Nesme}{Vincent Nesme}
\address[lab2]{Quantum information theory
\\ Universit\"at Potsdam
\\ Karl-Liebknecht-Str. 24/25
\\ 14476 Potsdam, Germany}
\email{vnesme@gmail.com} 
\urladdr{http://www.itp.uni-hannover.de/\~{}nesme/}


\keywords{cellular automata, neighborhood, quantum, block representation.}

\subjclass{F.1.1}

\begin{abstract}\noindent
We define the \block\ neighborhood of a reversible CA, which is related both to its decomposition 
into a product of block permutations and to quantum computing.  We give a purely combinatorial characterization of the \block\ neighborhood, which helps in two ways.  First, it makes the computation of the \block\ neighbourhood of a given CA relatively easy. Second, it allows us to derive upper bounds on the \block\ neighborhood: for a single CA as function of the classical and inverse neighborhoods, and for the composition of several CAs.  One consequence of that is a characterization of a class of ``elementary'' CAs that cannot be written as the composition of two simpler parts whose neighborhoods and inverse neighborhoods would be reduced by one half. 
\end{abstract}

\maketitle

\section*{Introduction}

Otherwise decent people have been known to consider reversible cellular automata (RCAs) and look for ways to decompose them into a product of reversible blocks permutations.  One big incentive for doing so is to ensure structural reversibility, as was the concern in \cite{margolus}, as it helps to design RCAs (see for instance \cite{MoritaCompUniv1D,MoritaCompUniv2D}), whereas determining from its local transition function whether a CA is reversible is undecidable \cite{undecidable}.  

Sadly, the relation is not clearly understood between both frameworks; several articles tackle this problem \cite{Kari_blocks,Kari_circuit_depth,durandlose}, whose conclusion, in a nutshell, is the following.  It is always possible, by increasing the size of the alphabet, to simulate a $d$-dimensional CA by a reversible block CA of depth at most $d+1$.  In the case of dimensions 1 and 2, up to shifts, no additional space and no coding is needed; it is still an open problem whether the same can be said in higher dimensions.

We will be here concerned with the size of the blocks, or rather, with the information on the neighborhood that is deducible purely geometrically from a block structure decomposition.

\begin{figure}[htbp]
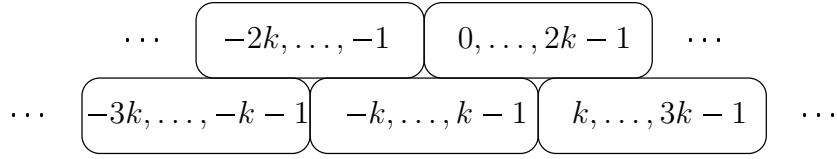
\caption{The geometric neighborhood in a block structure.}
\label{figure_blocks}
\end{figure}

If we just know that the CA is defined by such a structure, we can deduce, for instance, that the cell 0 has an influence only on the cells $-2k,\ldots,2k-1$, which means the neighborhood of this CA has to be included in $\interval{-2k+1}{2k}$.  But it is also true that the cells $-k$ and $k-1$ influence only the cells $-2k,\ldots,2k-1$, so the translation invariance tells us more: we can deduce that the neighborhood of this CA is included in $\interval{-k}{k}$.  Another way to look at it is to modify slightly the block structure, and update cell $0$ once and for good on the first step, so that the new structure would look something like Figure~\ref{figure_blocks2}.

\begin{figure}[htbp]
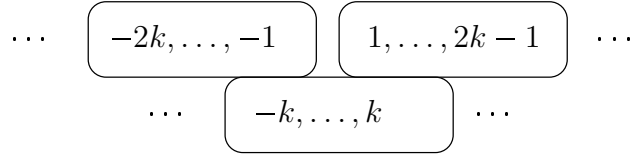
\caption{Block structure with a tooth gap.}
\label{figure_blocks2}
\end{figure}

If we concentrate only on the central block on the first line and the fact that no block of the second line acts on cell $0$ and ask what can then be the minimal size of the central block, we get to Figure~\ref{figure_blocks3} and our definition of the \block\ neighborhood in Definition~\ref{def_semilocal}.  Section~\ref{sec_def} is devoted to the basic properties of this neighborhood, in particular Proposition~\ref{old_definition} gives an expression of it in terms of combinatorics on words.

\begin{figure}[htbp]
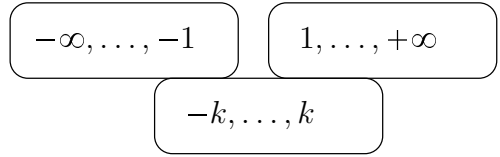
\caption{Simplified block structure.}
\label{figure_blocks3}
\end{figure}

So, how large must this central block be? Since it does all the work updating the state of cell $0$, it should at least include the neighborhood of this cell.  But there is a dual way to look at Figure~\ref{figure_blocks2} when it is turned upside-down.  What we know see is a block decomposition of the inverse CA, where the first step updates the complement of $\interval{-k}{k}$, so we also have a condition involving the neighborhood of the inverse CA, which has little to no relation to that of the CA itself.  Hence, there is something non trivial to say about that, and these considerations will be developed in Section~\ref{sec_main}, where the two results needed for
bounding block neighborhoods are stated.  Proposition \ref{borne_individuelle} is not new --- although it is the first time that it is put in direct relation with block decomposition --- but Corollary~\ref{corollaire_ridicule} and Proposition~\ref{mainthm} are.

A last caveat: while this article is written in a purely classical perspective, everything it deals with also has to do with quantum CAs (QCAs).  The definition of QCAs we are dealing with was introduced in \cite{Schumacher} as the natural extension of the usual definition of CAs to a universe ruled by reversible quantum laws.  It is founded on the same principles that rule usual CAs: discrete space-time, translation invariance, locality; in particular, QCAs have a similar notion of neighborhood.  It was already proven in that first article that reversible CAs can be naturally embedded into a quantum setting, turning them into QCAs.  However, curiously enough, the neighborhood of these QCAs --- the quantum neighborhood --- was not shown to be equal to that of the original CAs;  rather, a nontrivial bound was given (which is to be found as Proposition~\ref{borne_individuelle} of the present article).  It was then made explicit in \cite{ANW1} that the quantum neighborhood can indeed, and typically will, be strictly larger than the original one.

The authors tried to translate into purely classical terms a definition of the quantum neighborhood of quantized reversible CAs, and found the expression of Proposition~\ref{old_definition}, before realizing the close connection to block structures.  In retrospect, the link is hardly surprising, since a construction was given in \cite{ANW3} that uses auxiliary space to write a CA in a block structure, where each block acts on exactly the quantum neighborhood --- the construction is given in the quantum case, but applies to the classical case, mutatis mutandis.  Notions such as semicausality (Definition~\ref{def_semicausal}) and semilocalizability (Definition~\ref{def_semilocal}) are also imported from the quantum world, cf.~\cite{semicausal}.

So, in good conscience, the \block\ neighborhood could be called the quantum neighborhood, but since in the final version no explicit reference to the quantum model needs to be made, the name sounded a bit silly.  Nevertheless, if other natural neighborhoods were to be defined in relation to block structures, let it be said that the neighborhood we define and study in this article will always deserve ``quantum'' as a qualifier.

\section*{Notations}

\begin{itemize}
\item $\Sigma$ is the alphabet.
\item $a.b$ denotes the concatenation of words $a$ and $b$.
\item $a|_X$ is the restriction of word $a$ on a subset of indices $X$.
\item $a=b|_X$ means that words $a$ and $b$ coincide on $X$.
\item $\bar X$ denotes the complement of $X$, usually in $\Z$.
\item For $A,B\incl \Z$, $A+B$ is their Minkowski sum $\acco{a+b \mid a\in A , b\in B}$; similarly with $A-B$.
\item $\interval{x}{y}$ is the integer interval $[x;y]\cap \Z$.
\item $fg$ denotes the composition of CAs $f$ and $g$.
\item $\swap$ denotes the operation reversing the order in a tuple : $\swap (x_1,x_2,\ldots,x_n)=(x_n,x_{n-1},\ldots,x_1)$. It acts similarly on $\Sigma^\Z$ by $\swap(a)_n=a_{-n}$.
\end{itemize}

\section{Definitions}\label{sec_def}

\begin{defi}
For a bijection $f$ whose domain and range are written as products, its dual is defined by $\tilde f = \swap f^{-1} \swap$.  This applies in particular to the case where $f$ is a CA.  In this case $\tilde f$ is the conjugation of $f$ by the central symmetry.
\end{defi}

For instance, shifts are self-dual.  Clearly, $f\mapsto \tilde f$ is an involution.  In the remainder of this article, each time a notion (like a function or a property) is defined in term of a CA $f$, its dual, denoted by adding a tilde, is defined in the same way in term of $\tilde f$.

\begin{defi}
The (classic) neighborhood $\joliN(f)$ is the smallest subset $A$ of $\Z$ such that $v|_A$ determines $f(v)|_0$.
\end{defi}

The dual neighborhood $\tilde\joliN$ is thus defined by $\tilde \joliN(f)=\joliN(\tilde f)$.  The following two definitions are imported from \cite{semicausal}, where they are shown to be equivalent in the quantum case.

\begin{defi}\label{def_semicausal}
A function $f:A\times B\to C\times D$ is semicausal if its projection $C$ depends only on $A$, i.e. if there exists $g:A\to C$ such that $f(a,b)_C=g(a)$.
\end{defi}

\begin{defi}\label{def_semilocal}
A bijection $f:A\times B\to C\times D$ is (reversibly) semilocalizable if there exists a seevit $E$ and bijections $g:A\to C\times E$ and $h:D \to B\times E$ such that $f(a,b)= (g_C(a),\tilde h(g_E(a),b))$, as illustrated in Figure~\ref{semilocalization}.
\end{defi}

\begin{figure}[htbp]
\centering
\includegraphics[scale=.18]{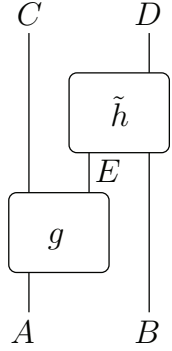}
\caption{a semilocalizable bijection}
\label{semilocalization}
\end{figure}

Semilocalizability is, as the name suggests, an asymmetric property, in the sense that applying a symmetry on Figure~\ref{semilocalization} along a vertical axis, i.e. swapping $A$ with $B$ and $C$ with $D$, breaks the semilocalizability.  However, a transformation that preserves the property is the central symmetry, which corresponds to taking the dual of $f$: the notion of semilocalizability is self-dual.

\begin{prop}\label{old_definition}
$f:A\times B\to C\times D$ is semilocalizable if and only if the three following conditions are met:
\begin{enumerate}
\item $f$ is semicausal;
\item $\tilde f$ is semicausal;
\item for every $a,a'\in A$ and $b,b'\in B$, if $f(a,b)_{D}=f(a',b)_{D}$ then $f(a,b')_{D}=f(a',b')_{D}$.
\end{enumerate}
\end{prop}

\begin{proof}
Suppose $f$ is semilocalizable.  Then obviously from Figure~\ref{semilocalization} both $f$ and $\tilde f$ are semicausal. Let $a,a'\in A$ and $b,b'\in B$ such that $f(a,b)_{D}=f(a',b)_{D}$.  That means $h^{-1}(b,g_E(a))=h^{-1}(b,g_E(a'))$ therefore $g_E(a)=g_E(a')$, and it follows immediately $f(a,b')_{D}=f(a',b')_{D}$.

Suppose now conditions (1), (2), (3) are met.  Let $\sim_A$ be the binary relation on $A$ defined by $a\sim_A a'$ iff $\forall b\in B\;f(a,b)=f(a',b)$.  Note that because of (3) this is equivalent to $\exists b\in B\;f(a,b)=f(a',b)$ (except if $B=\emptyset$, which is too trivial a case to worry about), from which we deduce 
\begin{equation}
\forall c,c'\in C\;\forall d\in D\quad \tilde f(d,c)\sim_A \tilde f(d,c') \label{existent}
\end{equation}

  It is clearly an equivalence relation, so using the fact that $f$ is semicausal we can define $g:A\to C\times (A/\sim_A)$ by $g(a)=(f(a,b)_C, [a])$, where $b$ is an arbitrary element of $B$ and $[a]$ is the class of $a$ in $A/\sim_A$.  One can define dually $\sim_D$ on $D$ and define $h:D\to B\times (D/\sim_D)$ by $h(d)=(\tilde f(d,c)_B, [d])$. 

It remains to be proven that $\alpha:\pa{\begin{array}{rcl} A/{\sim_A} & \to & D/{\sim_D} \\ \left[a\right] & \mapsto & \left[f(a,b)_D\right] \end{array}}$ is a well-defined bijection.  To prove that it is well-defined, we need to show that for every $a,a'\in A$ such that $a\sim_A a'$ and every $b,b'\in B$, $f(a,b)_D\sim_D f(a',b')_D$, which is easily done in two small steps.  First, by definition of $\sim_A$, $f(a,b)_D=f(a',b)_D$.  Then, by the dual of (\ref{existent}), $f(a',b)_D\sim_D f(a',b')$.  We now prove that $\alpha$ is bijection by showing that its inverse is its dual, defined by $\tilde\alpha([d]) = [\tilde f(d,c)_A]$, so that $\tilde\alpha\alpha ([a])=[\tilde f(f(a,b)_D,c)_A]$.  Since this value is independent of $c$, we can try in particular with $c=f(a,b)_C$, where it is clear that we get $[a]$.
\end{proof}

\begin{defi}
For a CA $f$ on the alphabet $\Sigma$ and $X,Y$ two subsets of $\Z$, let \ppte{X}{Y}{f} be the property: ``$f$ seen as a function from $\Sigma^X\times \Sigma^{\bar X}$ to $\Sigma^Y\times \Sigma^{\bar Y}$ is semilocalizable''.
\end{defi}

Some property are obvious from the definition of semilocalizability, especially from Figure~\ref{semilocalization}.  Let us give two basic examples.

\begin{lemma}\label{welldefined}
If \ppte{X}{Y}{f} holds, then so does \ppte{X'}{Y'}{f} for every $X'\supseteq X$ and $Y'\subseteq Y$.
\end{lemma}

\begin{itemize}
\item For a CA seen as a function from $\Sigma^X\times \Sigma^{\bar X}$ to $\Sigma^Y\times \Sigma^{\bar Y}$, being semicausal means $X\supseteq Y+\joliN(f)$; the semicausality of $\tilde f$ means $X\supseteq Y + \tilde\joliN(f)$.
\end{itemize}

The following property, however, is easier to prove with Proposition~\ref{old_definition} in mind.

\begin{lemma}\label{intersection}
If \ppte{X}{Y}{f} and \ppte{X'}{Y}{f}, then \ppte{X\cap X'}{Y}{f}.
\end{lemma}

\begin{proof}
Let $a,b$ be words on $X\cap X'$, and $u,v$ words on $\overline{X\cap X'}$, and suppose $f(a.u)=f(b.u)|_{\bar Y}$.  Let $u'$ be the word on $\overline{X\cap X'}$ that is equal to $u$ on $X\setminus X'$, $v$ elsewhere.  According to \ppte{X}{Y}{f}, $f(a.u')=f(b.u')|_{\bar Y}$; we then conclude from \ppte{X'}{Y}{f} that $f(a.v)=f(b.v)|_{\bar Y}$.
\end{proof}

Not that we get immediately the following corollary from the selfduality of semilocalizability: if \ppte{X}{Y}{f} and \ppte{X}{Y'}{f}, then \ppte{X}{Y\cup Y'}{f}.

We have now established all the properties on \ppte{X}{Y}{f} required to define the \block\ neighborhood.

\begin{defi}\label{carac_quantic}
The \block\ neighborhood $\voisq(f)$ of $f$ is the smallest $X$ such that \ppte{X}{\acco{0}}{f} holds.
\end{defi}

The word ``neighborhood'' is not gratuitous.  In fact, \ppte{X}{Y}{f} behaves exactly like ``$X$ includes the neighborhood of $Y$ for some CA $f'$ such that $\joliN(f')=\voisq(f)$'', as stated in the next lemma.  The idea is that $\voisq$ characterizes a notion of dependency that is very similar to the usual one characterized by $\joliN$.

\begin{lemma}\label{xyz}
\ppte{X}{Y}{f} is equivalent to $X\supseteq Y+\voisq(f)$.
\end{lemma}

\begin{proof}
Suppose $X\supseteq Y+\voisq(f)$.  By translation invariance, for all $y\in Y$, we have \ppte{\acco{y}+\voisq(f)}{\acco{y}}{f}, which, according to Lemma~\ref{welldefined}, implies \ppte{X}{\acco{y}}{f}.  Invoking now the dual of Lemma~\ref{intersection}, we get \ppte{X}{Y}{f}.

For the reciprocal, suppose now \ppte{X}{Y}{f}.  According to Lemma~\ref{welldefined}, we have, for every $y\in Y$, \ppte{X}{\acco{y}}{f}; but that is, by definition of $\voisq$, equivalent to $X\supseteq \acco{y}+\voisq(f)$, so we must have $X\supseteq Y+\voisq(f)$.
\end{proof}

So, the \block\ neighborhood is just one kind of neighborhood.  However, $\voisq$ has by contruction one property that $\joliN$ does not share: it is self-dual.  It is not enough to make interesting, and many questions are left open at this point.  We just know $\voisq\supseteq\joliN\cup\tilde\joliN$, but do we have a lower bound on $\voisq$? Is it always finite? The answer is in Corollary~\ref{corollaire_ridicule}. How do the \block\ neighborhoods compose? It can be easily inferred from Figure~\ref{semilocalization} that $\voisq(gf)\incl \voisq(g)+\voisq(f)$, which is certainly good news, but Proposition~\ref{mainthm} provides much more interesting bounds.

\section{Main theorem}\label{sec_main}

\begin{proposition}\label{borne_individuelle}
$\voisq\incl\joliN - \joliN + \tilde \joliN$.
\end{proposition}
Sanity check: $\joliN - \joliN + \tilde \joliN$ contains indeed $\joliN$ because $\joliN \cap \tilde \joliN\neq \emptyset$, which follows from $\acco{0}= \joliN(f f^{-1}) \incl \joliN(f)-\tilde\joliN(f)$.

\begin{proof}
The proof of this proposition can be essentially found in \cite{Schumacher}, where the result is stated as Lemma~4, albeit in a foreign formalism.  Another avatar of the proposition and its proof can be found in the form of Lemma~3.2 of \cite{JAC}.  In order to keep this article self-contained, we give yet another proof.

Let us then prove \ppte{\joliN(f) - \joliN(f) + \tilde \joliN (f)}{\acco{0}}{f}.  Let $a,b$ be words on $\joliN(f) - \joliN(f) + \tilde \joliN (f)$ and $u,v$ words on its complement such that $f(a.u)=f(b.u)|_{\overline{\acco{0}}}$.   Applying $f^{-1}$ to that equality we get $a.u=b.u|_{\overline{\tilde\joliN(f)}}$, which implies of course $a.v=b.v|_{\overline{\tilde\joliN(f)}}$, from which we obtain $f(a.v)=f(b.v)|_{\overline{-\joliN(f)+\tilde\joliN(f)}}$.  On the other hand, $f(w)|_{-\joliN(f)+\tilde\joliN(f)}$ is a function of $w|_{\joliN(f) - \joliN(f) + \tilde \joliN (f)}$, so $f(a.u)=f(a.v)|_{-\joliN(f)+\tilde\joliN(f)}$ and $f(b.u)=f(b.v)|_{-\joliN(f)+\tilde\joliN(f)}$, which in the end proves $f(a.v)=f(b.v)|_{\overline{\acco{0}}}$.
\end{proof}

Rather surprisingly, this bound is not self-dual, which allows us to reinforce it immediately.

\begin{corollary}\label{corollaire_ridicule}
$\voisq\incl(\joliN - \joliN + \tilde \joliN) \cap (\tilde\joliN - \tilde\joliN + \joliN)$.
\end{corollary}

\begin{proposition}\label{mainthm}
Let $f_1,\ldots f_n$ be reversible CAs. Then
$$\voisq(f_n\cdots f_1)\incl \bigcup\limits_{k=1}^n \pa{\tilde \joliN(f_n\cdots f_{k+1}) + \voisq(f_k) + \joliN(f_{k-1} \cdots f_1)}.$$
\end{proposition}

This formula could seem at first glance not to be self-dual, and therefore obviously suboptimal, but there is more to the duality than just putting and removing tildes.  Since $\widetilde{fg}=\tilde g \tilde f$, we have $\widetilde{\voisq}(f_n\cdots f_1)= \voisq(\tilde f_1\cdots\tilde f_n)$; it is from here straightforward to check that the formula is indeed self-dual.

\begin{proof}
Let $\joliV=\bigcup\limits_{k=1}^n \pa{\tilde \joliN(f_n\cdots f_{k+1}) + \voisq(f_k) + \joliN(f_{k-1} \cdots f_1)}$; we have to prove \ppte{\joliV}{\acco{0}}{f}.  Let $a,b$ be words on $\joliV$, $u,v$ words on $\bar\joliV$, and assume $f_n\cdots  f_1(a.u)= f_n\cdots  f_1(b.u)|_{\overline{\acco{0}}}$.

For $k\in\interval{0}{n}$, let $\joliC_k=\tilde\joliN(f_n\cdots f_{k+1})$;  for $k\in\interval{1}{n}$, let $\joliK_k=\joliC_k+\voisq(f_k)$ and $\joliD_k=\joliN(f_{k-1}\cdots f_{1})$.  For $k\in\interval{1}{n}$, let $\joliV_k=\joliK_k+ \joliD_k$; by definition, $\joliV=\bigcup\limits_{k=1}^n \joliV_k$.

We will prove by induction the following hypothesis \rec{k} for $k\in \interval{0}{n}$:

\begin{itemize}
\item $f_{k}\cdots f_1(a.u)=f_{k}\cdots f_1(b.u)|_{\overline{\joliC_k}}$ and
\item $f_{k}\cdots f_1(a.v)=f_{k}\cdots f_1(b.v)|_{\overline{\joliC_k}}$.
\end{itemize}

Since we already know $a.u=b.u|_{\overline{\joliC_0}}$, it follows immediately $a.v=b.v|_{\overline{\joliC_0}}$, so \rec{0} is true.

Suppose \rec{k} for some $k\in\interval{0}{n-1}$.  Let $a'=f_{k}\cdots  f_1(a.u)|_{\joliK_{k+1}}$ and $b'=f_{k}\cdots  f_1(b.u)|_{\joliK_{k+1}}$; since $\joliK_{k+1}+\joliD_{k+1}\incl \joliV$, $a'$ and $b'$ are respectively equal to $f_{k}\cdots  f_1(a.v)|_{\joliK_{k+1}}$ and $f_{k}\cdots  f_1(b.v)|_{\joliK_{k+1}}$.  Let us define $u'=f_{k}\cdots  f_1(a.u)|_{\overline{\joliK_{k+1}}}$ and $v'=f_{k}\cdots  f_1(a.v)|_{\overline{\joliK_{k+1}}}$.  We have 

$$\joliC_{k}=\tilde\joliN(f_n\cdots f_{k+1})
\incl \tilde\joliN(f_n\cdots f_{k+2}) + \tilde\joliN(f_{k+1})\incl \joliC_{k+1}+ \voisq(f_{k+1}) = \joliK_{k+1}.$$

We can therefore deduce from \rec{k} that $u'$ and $v'$ are respectively equal to $f_{k}\cdots  f_1(b.u)|_{\overline{\joliK_{k+1}}}$ and $f_{k}\cdots  f_1(b.v)|_{\overline{\joliK_{k+1}}}$.  By definition of $\joliC_{k+1}$, since $f_n\cdots  f_1(a.u)= f_n\cdots  f_1(b.u)|_{\overline{\acco{0}}}$, we have $f_k\cdots  f_1(a.u)= f_k\cdots  f_1(b.u)|_{\overline{\joliC_k}}$, which is the first point of \rec{k+1}.  Since $\joliK_{k+1}=\joliC_{k+1}+\voisq(f_{k+1})$, according to Lemma~\ref{xyz}, we have \ppte{\joliK_{k+1}}{\joliC_{k+1}}{f_{k+1}}.  We therefore deduce the second point of \rec{k+1}.

So in the end we get \rec{n}, which concludes the proof because $\joliC_n=\acco{0}$.
\end{proof}

\begin{corollary}\label{corollaire_itere}
Suppose $\joliN(f)\incl\interval{-\alpha}{\beta}$ and $\tilde\joliN(f)\incl\interval{-\gamma}{\delta}$. Then $\voisq(f^k)\incl \interval{-(k+1)\max\pa{\alpha,\gamma}-\min\pa{\beta,\delta}}{\pa{k+1}\max\pa{\beta,\delta}+\min\pa{\alpha,\gamma}}$.
\end{corollary}

For $X\incl\R$, let $X^*$ be its convex hull and for $\lambda\in\R$, $\lambda X=\acco{\lambda x \mid x\in X}$.  We get, for any reversible CA, the asymptotic relation $\lim\limits_{k\to +\infty}\frac{1}{k}\voisq(f^k)^*\incl \joliN(f) ^*\cup\tilde\joliN(f)^*$.  Let us assume we are in the case $\lim\limits_{k\to +\infty}\frac{1}{k}\joliN(f^k)^*= \joliN(f) ^*$ and  $\lim\limits_{k\to +\infty}\frac{1}{k}\tilde\joliN(f^k)^*= \tilde \joliN(f) ^*\cup\tilde\joliN(f)^*$.  Then what this means informally is that condition~(3) in Proposition~\ref{old_definition} applied to $f^k$ becomes less restrictive as $k$ grows, and fades at the limit.

It is interesting to note the relation with Kari's constructions in \cite{Kari_blocks} and \cite{Kari_circuit_depth}.  We will briefly discuss the latter; it is of course stated in dimension~2, but that is not an obstacle to comparison, as the same construction can be made in dimension 1, or our analysis generalized to dimension~2 (cf. section~\ref{sec_general}).  Let us place ourselves in dimension~1.  Let $f$ be a CA whose neighborhood and dual neighborhood are both included in $\interval{-1}{1}$.  In this case, Corollary~\ref{corollaire_itere} implies $\voisq(f^k)\incl\interval{-(k+2)}{k+2}$.  Kari proves that there is an embedding $\varphi$ and a CA $g$ such that $f=\varphi g \varphi^{-1}$, where $g$ fulfills by construction $\voisq(g)\incl\interval{-1}{1}$.  Kari's construction therefore contains an asymptotically optimal bound on $\voisq(f^k)$.

\begin{corollary}\label{indecomposable}
If the neighborhoods and dual neighborhoods of $f$ and $g$ are included in $\interval{-n}{n}$, then $\voisq(fg)\incl\interval{-4n}{4n}$.
\end{corollary}

The contraposition is actually more interesting.  Consider $h$, whose neighborhood and dual neighborhood are both included in $\interval{-n}{n}$; its \block\ neighborhood has to be contained in $3\interval{-n}{n}$.  It seems perfectly reasonable to assume that $h$ could be a composition of two more elementary reversible cellular automata $f$ and $g$ having strictly smaller neighborhoods, containing $\frac{1}{2}\interval{-n}{n}$ but close to it.  Actually, if no restriction is imposed on the behaviour of $f^{-1}$ and $g^{-1}$, maybe even allowing $f$ and $g$ to be nonreversible, it is certainly possible to decompose $h$ in such a way by increasing the size of the alphabet.  However, if the dual neighborhoods of $f$ and $g$ are also required to be close to $\frac{1}{2}\interval{-n}{n}$, then such a decomposition will not be possible if $\voisq(h)$ is too large.  For instance, if $\voisq(h)$ is not contained in $\frac{5}{2}\interval{n}{n}$, then $\joliN(f)$, $\tilde\joliN(f)$, $\joliN(g)$ and $\tilde\joliN(g)$ cannot all be included in $\frac{5}{8}\interval{-n}{n}$.  In this sense, $h$ can be considered ``elementary''.

\section{Remarks}

We gather in this section several unrelated observations about the \block\ neighborhood.

\subsection{Subtraction Automata}

Suppose $\Sigma$ can be provided with a binary operation $\cdot - \cdot$ such that:

\begin{itemize}
\item there exists an element of $\Sigma$ denoted $0$ such that $x=y$ is equivalent to $x-y=0$;
\item $f$ is an endomorphism of $(\Sigma^\Z,-)$, where $-$ is defined component-wise on $\Sigma^\Z$.
\end{itemize}

We say in this case $f$ admits a subtraction.  For instance, linear automata as defined in \cite{fractal} admit subtractions.

\begin{proposition}
Automata with subtractions have minimal \block\ neighborhoods.
In other words, for any automaton $f$ admitting a subtraction, $\voisq(f)=\joliN(f)\cup\tilde\joliN(f)$.
\end{proposition}

\proof
Let $A$ be any subset of $\Z$, $a,b$ be words on $A$, $u,v$ words on $\bar A$, and suppose $f(a.u)=f(b.u)$.  Then $f(a.v)-f(b.v)=f(a.v-b.v)= f((a-b).0)=f(a.u-b.u)=f(a.u)-f(b.u)=0$.
\qed

\subsection{Generalization}\label{sec_general}

We can actually drop many properties of the CAs that are irrelevant to the notions developed in this article.  We don't need translation invariance.  We don't need the alphabet to be finite.  We don't need the neighborhoods to be finite.  We don't need the domain and range cell structures to be identical.  In this abstract setting, a ``reversible automaton'' is a bijection from $\prod\limits_{i\in I} X_i$ to $\prod\limits_{j\in J} Y_j$ and $\joliN(f)$ is a function from $\Pa(J)$ to $\Pa(I)$ which to $B\incl J$ associates the minimal subset $A$ of $I$ such that $f(x)|_B$ depends only on $x|_A$.  In general, a function $\alpha: \prod\limits_{i\in I} X_i\to \prod\limits_{j\in J} Y_j$ is a \emph{neighborhood scheme} if for all $Y$, $\alpha(Y)=\bigcup_{X\incl Y} \alpha(X)$; $\joliN(f)$ is of course one example of a neighborhood scheme.  The usual definition of the neighborhood in the case of a cellular automaton corresponds here to $\joliN(f)(\acco{0})$.
Any function $\alpha:\Pa(J)\to\Pa(I)$ has a transpose $\alpha^{\dagger}:\Pa(I)\to\Pa(J)$ defined by $\alpha^{\dagger}(A)$ being the largest subset $B$ of $J$ such that $\alpha(B)\incl A$.  We have indeed $(\alpha^\dagger)^\dagger=\alpha$, and for usual one-dimensional CAs, $^\dagger$ corresponds to $\joliN\mapsto -\joliN$.

There is not anymore any good notion of duality on automata, but $\tilde\joliN(f)$ can be defined as $\joliN^\dagger (f^{-1})$.  Of course the definition of $\voisq(f)$ cannot make any reference to $0$, instead $\voisq(f)(B)$ is now the smallest subset of $I$ fulfilling \ppte{A}{B}{f}.  The self-duality of $\voisq$ is of course still valid.  Lemmas \ref{welldefined} and \ref{intersection} state respectively that $\voisq$ is well-defined and that it is a neighborhood scheme.

Lemma~\ref{xyz} (used once at the end of the proof of Proposition~\ref{mainthm}) becomes ``\ppte{X}{Y}{f} is equivalent to $X\supseteq \bigcup\limits_{y\in Y}\voisq(f)(Y)$'', which is precisely the definition of $\voisq$; it can therefore be forgotten, as the triviality it is now.  Proposition~\ref{borne_individuelle} becomes $\voisq\incl \joliN \circ \joliN^\dagger\circ \tilde\joliN$, Corollary~\ref{corollaire_ridicule} changes accordingly, and Proposition~\ref{mainthm} remains true when ``$+$'' is substituted with ``$\circ$''.  It follows that indecomposability results such as Corollary~\ref{indecomposable} are extremely robust:  they cannot be overcome by increasing the size of the alphabet or relaxing the translational invariance.  It also shows of course the limitations of this method, namely that it is utterly unable to exploit these parameters.

\subsection{Optimality}\label{sec_optimal}

The bounds presented in Corollary~\ref{corollaire_ridicule} and Proposition~\ref{mainthm} seem  peculiar enough as to be suspect of non-optimality.  However, we have been unable to come up with a better approximation, and would rather tend to think that they cannot be improved.  We will concentrate on Corollary~\ref{corollaire_ridicule} alone, whose optimality is conjectured in the following statement.

\begin{conjecture}\label{conjecture}
For any subsets $X,Y$ and $Z$ of $\Z$ such that $X\cup Y\incl Z\incl (X-X+Y)\cap (Y-Y+X)$, if there exists a CA $f$ such that $\joliN(f)=X$ and $\tilde\joliN(f)=Y$, then there exists a CA $g$ such that $\joliN(g)=X$, $\tilde\joliN(g)=Y$ and $\voisq(g)=Z$.
\end{conjecture}

This section will be devoted to proving the following weaker version:

\begin{prop}\label{weak_prop}
Conjecture~\ref{conjecture} is true when $Z\incl \acco{2y-x\mid x,y\in X\cap Y}$.
In particular it is true if $X$ and $Y$ are equal intervals.
\end{prop}

\begin{proof}
Given that there is by hypothesis a CA $f$ such that $\joliN(f)=X$ and $\tilde\joliN(f)=Y$, we only need to prove that for every $z\in Z$ there exists a CA $g_z$ such that $\joliN(g_z)\incl X$, $\tilde\joliN(g_z)\incl Y$ and $z\in\joliN(g_z)\incl Z$.  Then the proposition is proven by considering the direct sum of $f$ and all these $g_z$'s.




 The Toffoli automaton presented in \cite{ANW1} (definition~12), defined by $\Sigma=\pa{\Z/2\Z}^2$ and $T(v)_0=\pa{v_0^2+v_0^1v_1^1,v_1^1}$, will serve as the basic constructing tool for $g_z$.  Its inverse is given by $T^{-1}(v)_0=\pa{v_{-1}^2,v_0^1+v_{-1}^2v_0^2}$, so we clearly have $\joliN(T)=\tilde\joliN(T)=\acco{0;1}$.  Let us prove $\voisq(T)=\interval{0}{2}$, by proving first that \ppte{\N}{\N}{T} is true, and then that \ppte{\acco{0;1}}{\acco{0}}{T} is false.  Let then $a,b$ be words on $\N$ and $u,v$ words on its complement, and suppose $T(a.u)=T(b.u)_{\bar\N}$.  In particular, $T(a.u)_{-1}^2= T(b.u)_{-1}^2$, which implies $a_0^1=b_0^1$, so we get immediately $T(a.v)=T(b.v)_{\bar\N}$, which proves \ppte{\N}{\N}{T}.  Consider now the words $a=(0,0)(0,0)$ and $b=(0,0)(1,1)$ on $\acco{0;1}$, and $u$, $v$ the words on its complement that are $(0,0)$ everywhere except in position $2$, where $u_2=(1,0)$.  We have $T(a.u)=T(b.u)_{\overline{\acco{0}}}$ but $T(a.v)_1=(0,0)$ while $T(b.v)_1=(1,0)$, therefore \ppte{\acco{0;1}}{\acco{0}}{T} is false.

This CA can be obviously expanded into an automaton $T_l$ such that $\joliN (T_l)=\tilde\joliN (T_l)=\acco{0;l}$ and $\voisq(T_l)=\acco{0;l;2l}$. 

More generally, for any nonempty intervals $X$ and $Z$ of $\Z$ such that $X\incl Z\incl X-X+X$, there is a CA $f$ such that $\joliN(f)=\tilde\joliN(f)=X$ and $\voisq(f)=Z$.  We can engineer such an $f$ by considering the direct sum of several CAs.  First, for each element $x\in X$, consider the shift by $-x$: the sum of all these shifts is a CA $g$ such that $\joliN(g)=\tilde\joliN(g)=\voisq(g)=X$.   Then, for each element $z\in Z$, choose $x$ and $y$ in $X\cap Y$ such that $z=2y-x$.  The automaton $g_z=\sigma^x T_{y-x}$, where $\sigma$ is the elementary shift to the left, is then such that $\joliN (T_l)=\tilde\joliN (T_l)=\acco{x;y}$ and $\voisq(T_l)=\acco{x;y;z}$, which concludes the proof.
\end{proof}

\section*{Conclusion}

Of course a lot of questions remain.  Are the upper bounds on the \block\ neighborhood given in this article optimal under all circumstances?  And is it possible to make these bounds more efficient by including as parameters the size of the alphabet and the requirement that the transformations be translation invariant?  This would probably require a whole different technique.

Something happened in this article that is increasingly common: after a theory grows a quantum extension (in this case QCAs join the family of CAs) and new tools and techniques are invented to study the quantum setup, they come back to the classical setup (semilocalizability comes to mind, and a lot of others are disguised as combinatorial properties) and bring various insights, simpler proofs and/or new results.

The \block\ neighborhood is nothing else than the quantum neighborhood.  It shows what had been grasped until then only intuitively: whereas CAs can be defined by their local transition functions, QCAs are intrisically block-structured.  In that sense, working on QCAs is a lot like working on CAs with a restricted bag of tools that includes only local permutations --- duplication or destruction of information are stricly forbidden.  It also means that, even staying in a purely classical framework, finding this kind of constructions is worthwhile and meaningful, even in the case where a result is already known to be attainable by another method.  Not only will the construction be nicer in a purely abstract way, because it will employ only elementary means: it will also have the benefit of being immediately transposable to the quantum case.

\section*{Acknowledgements}
  The authors would like to thank Jarkko Kari for showing them, to their amazement, how the neighborhood and the inverse neighborhood of CAs depend so little on each other, even in the iterated dynamics.  They also gratefully acknowledge the support of the Deutsche Forschungsgemeinschaft (Forschergruppe 635) and the EU (project QICS).

\bibliographystyle{alpha}
\bibliography{biblio}

\end{document}

%% file: kkmacros.tex
\usepackage{verbatim,color,mathabx,cyrillic}
\usepackage{graphicx}

\newcommand{\acco}[1]{\left\{#1\right\}}

\newcommand{\pa}[1]{\left(#1\right)}

\newcommand{\interval}[2]{\left[\hspace{-1ex}\left[\hspace{0.5ex}#1;#2 \hspace{0.5ex} \right]\hspace{-1ex}\right]}

\newcommand{\ppte}[3]{$\text{\bf Q}_{#1}^{#2}(#3)$}
\newcommand{\rec}[1]{($\mathcal{H}_{#1}$)}
\newcommand{\swap}{\bigvarstar}

\newcommand{\block}{block}

\newcommand{\voisq}{\mathcal{BN}}

\newcommand{\incl}{\subseteq}

\newcommand{\joliC}{\mathcal{C}}
\newcommand{\joliD}{\mathcal{D}}

\newcommand{\joliN}{\mathcal{N}}
\newcommand{\joliV}{\mathcal{V}}
\newcommand{\joliK}{\mathcal{K}}

\newcommand{\N}{\mathbb{N}}
\newcommand{\R}{\mathbb{R}}
\newcommand{\Pa}{\mathcal{P}}
\newcommand{\Z}{\mathbb{Z}}

%% file: blocks.pdf_t
\begin{picture}(-27,0)%
\includegraphics{blocks}%
\end{picture}%
\setlength{\unitlength}{4144sp}%
\begingroup\makeatletter\ifx\SetFigFont\undefined%
\gdef\SetFigFont#1#2#3#4#5{%
  \reset@font\fontsize{#1}{#2pt}%
  \fontfamily{#3}\fontseries{#4}\fontshape{#5}%
  \selectfont}%
\fi\endgroup%
\begin{picture}(4710,924)(-14,-2773)
\put(1981,-2581){\makebox(0,0)[lb]{\smash{{\SetFigFont{12}{14.4}{\rmdefault}{\mddefault}{\updefault}{\color[rgb]{0,0,0}$-k, \ldots, k-1$}%
}}}}
\put(3331,-2581){\makebox(0,0)[lb]{\smash{{\SetFigFont{12}{14.4}{\rmdefault}{\mddefault}{\updefault}{\color[rgb]{0,0,0}$k, \ldots, 3k-1$}%
}}}}
\put(4006,-2131){\makebox(0,0)[lb]{\smash{{\SetFigFont{12}{14.4}{\rmdefault}{\mddefault}{\updefault}{\color[rgb]{0,0,0}$\cdots$}%
}}}}
\put(451,-2581){\makebox(0,0)[lb]{\smash{{\SetFigFont{12}{14.4}{\rmdefault}{\mddefault}{\updefault}{\color[rgb]{0,0,0}$-3k, \ldots, -k-1$}%
}}}}
\put(1261,-2131){\makebox(0,0)[lb]{\smash{{\SetFigFont{12}{14.4}{\rmdefault}{\mddefault}{\updefault}{\color[rgb]{0,0,0}$-2k, \ldots,  -1$}%
}}}}
\put(2656,-2131){\makebox(0,0)[lb]{\smash{{\SetFigFont{12}{14.4}{\rmdefault}{\mddefault}{\updefault}{\color[rgb]{0,0,0}$0,\ldots, 2k-1$}%
}}}}
\put(4681,-2581){\makebox(0,0)[lb]{\smash{{\SetFigFont{12}{14.4}{\rmdefault}{\mddefault}{\updefault}{\color[rgb]{0,0,0}$\cdots$}%
}}}}
\put(676,-2131){\makebox(0,0)[lb]{\smash{{\SetFigFont{12}{14.4}{\rmdefault}{\mddefault}{\updefault}{\color[rgb]{0,0,0}$\cdots$}%
}}}}
\put(  1,-2581){\makebox(0,0)[lb]{\smash{{\SetFigFont{12}{14.4}{\rmdefault}{\mddefault}{\updefault}{\color[rgb]{0,0,0}$\cdots$}%
}}}}
\end{picture}%

%% file: blocks2.pdf_t
\begin{picture}(-30,0)%
\includegraphics{blocks2}%
\end{picture}%
\setlength{\unitlength}{4144sp}%
\begingroup\makeatletter\ifx\SetFigFont\undefined%
\gdef\SetFigFont#1#2#3#4#5{%
  \reset@font\fontsize{#1}{#2pt}%
  \fontfamily{#3}\fontseries{#4}\fontshape{#5}%
  \selectfont}%
\fi\endgroup%
\begin{picture}(3495,924)(526,-2773)
\put(4006,-2131){\makebox(0,0)[lb]{\smash{{\SetFigFont{12}{14.4}{\rmdefault}{\mddefault}{\updefault}{\color[rgb]{0,0,0}$\cdots$}%
}}}}
\put(2656,-2131){\makebox(0,0)[lb]{\smash{{\SetFigFont{12}{14.4}{\rmdefault}{\mddefault}{\updefault}{\color[rgb]{0,0,0}$1,\ldots, 2k-1$}%
}}}}
\put(1126,-2131){\makebox(0,0)[lb]{\smash{{\SetFigFont{12}{14.4}{\rmdefault}{\mddefault}{\updefault}{\color[rgb]{0,0,0}$-2k, \ldots,  -1$}%
}}}}
\put(541,-2131){\makebox(0,0)[lb]{\smash{{\SetFigFont{12}{14.4}{\rmdefault}{\mddefault}{\updefault}{\color[rgb]{0,0,0}$\cdots$}%
}}}}
\put(1981,-2581){\makebox(0,0)[lb]{\smash{{\SetFigFont{12}{14.4}{\rmdefault}{\mddefault}{\updefault}{\color[rgb]{0,0,0}$-k, \ldots, k$}%
}}}}
\put(1351,-2581){\makebox(0,0)[lb]{\smash{{\SetFigFont{12}{14.4}{\rmdefault}{\mddefault}{\updefault}{\color[rgb]{0,0,0}$\cdots$}%
}}}}
\put(3286,-2581){\makebox(0,0)[lb]{\smash{{\SetFigFont{12}{14.4}{\rmdefault}{\mddefault}{\updefault}{\color[rgb]{0,0,0}$\cdots$}%
}}}}
\end{picture}%

%% file: blocks3.pdf_t
\begin{picture}(0,0)%
\includegraphics{blocks3}%
\end{picture}%
\setlength{\unitlength}{4144sp}%
\begingroup\makeatletter\ifx\SetFigFont\undefined%
\gdef\SetFigFont#1#2#3#4#5{%
  \reset@font\fontsize{#1}{#2pt}%
  \fontfamily{#3}\fontseries{#4}\fontshape{#5}%
  \selectfont}%
\fi\endgroup%
\begin{picture}(2904,924)(934,-2773)
\put(2656,-2131){\makebox(0,0)[lb]{\smash{{\SetFigFont{12}{14.4}{\rmdefault}{\mddefault}{\updefault}{\color[rgb]{0,0,0}$1,\ldots, +\infty$}%
}}}}
\put(1081,-2131){\makebox(0,0)[lb]{\smash{{\SetFigFont{12}{14.4}{\rmdefault}{\mddefault}{\updefault}{\color[rgb]{0,0,0}$-\infty, \ldots,  -1$}%
}}}}
\put(1981,-2581){\makebox(0,0)[lb]{\smash{{\SetFigFont{12}{14.4}{\rmdefault}{\mddefault}{\updefault}{\color[rgb]{0,0,0}$-k, \ldots, k$}%
}}}}
\end{picture}%